\documentclass[letterpaper,11pt]{article} 
\usepackage[margin=1in]{geometry}
\usepackage{helvet}  
\usepackage{courier}  
\usepackage[hyphens]{url}  
\usepackage{graphicx} 
\usepackage{natbib}  
\usepackage{caption} 

\usepackage{amssymb}
\usepackage{amsthm}
\usepackage{csquotes}
\usepackage{caption}
\usepackage{amssymb}
\usepackage{amsfonts}
\usepackage[usenames,dvipsnames]{color}
\usepackage{xcolor}
\usepackage{amsmath}
\usepackage{stackrel}
\usepackage{mathtools}
\usepackage{comment}
\usepackage{natbib}

\usepackage[ruled,linesnumbered,lined]{algorithm2e}
\usepackage[bibliography=common]{apxproof}

\newtheoremrep{theorem}{Theorem}[section]
\newtheoremrep{lemma}[theorem]{Lemma}
\newtheorem*{lemma*}{Lemma}
\newtheoremrep{claim}[theorem]{Claim}
\newtheoremrep{proposition}[theorem]{Proposition}

\newtheorem{corollary}[theorem]{Corollary}

\theoremstyle{definition}
\newtheorem{definition}[theorem]{Definition}
\newtheorem{example}[theorem]{Example}

\newenvironment{claimproof}[1]{\noindent\emph{Proof.}\space#1}
{\hfill $\blacksquare$\medskip\par}

\newcommand{\vmms}{\mu}

\newcommand{\inst}{\mathcal{I}}

\newcommand{\bundle}{B} 

\newcommand{\alloc}{A} 
\newcommand{\allocs}{\mathbf \alloc} 
\newcommand{\alloci}[1][i]{\alloc_{#1}} 

\newcommand{\items}{\mathcal{M}} 

\newcommand{\vai}[2][i]{v_{#1}\left(#2\right)} 

\newcommand{\last}{\omega}

\newcommand{\swap}[3]{\textit{SWAP}\left(#1,#2,#3\right)}

\newcommand{\lex}[1]{\stackrel{\mathclap{\tiny lex}}{#1}}

\title{Improved Maximin Share Approximations for Chores\\ by Bin Packing}

\author{ 
Jugal Garg\thanks{University of Illinois at Urbana Champaign, USA} 
\\ \tt{\small jugal@illinois.edu}
\and
Xin Huang\thanks{Kyushu University, Fukuoka, Japan}
\\\tt{\small huangxin@inf.kyushu-u.ac.jp} 
\and
Erel Segal-Halevi\thanks{Ariel University, Ariel 40700, Israel}
\\\tt {\small erelsgl@gmail.com}
}
\date{}

\begin{document}

\maketitle
\begin{abstract}
We study fair division of indivisible chores among $n$ agents with additive cost functions using the popular fairness notion of maximin share (MMS). Since MMS allocations do not always exist for more than two agents, the goal has been to improve its approximations and identify interesting special cases where MMS allocations exists.
We show the existence of 
\begin{itemize}
\item 1-out-of-$\lfloor \frac{9}{11}n\rfloor$ MMS allocations, which improves the state-of-the-art factor of 1-out-of-$\lfloor \frac{3}{4}n\rfloor$. 
\item MMS allocations for factored instances, which resolves an open question posed by Ebadian et al. (2021).  
\item $15/13$-MMS allocations for personalized bivalued instances, improving the state-of-the-art factor of $13/11$. 
\end{itemize}
We achieve these results by leveraging the HFFD algorithm of Huang and Lu (2021). Our approach also provides polynomial-time algorithms for computing an MMS allocation for factored instances and a $15/13$-MMS allocation for personalized bivalued instances. 
\end{abstract}

\section{Introduction}
Fair division of indivisible tasks (or chores) has garnered significant attention recently due to its applications in various multi-agent settings; see recent surveys~\citep{amanatidis2023fair,liu2024mixed}. The problem is to find a \emph{fair} partition of a set $\items$ of $m$ indivisible chores among $n$ agents with preferences. We assume that each agent $i$ has additive preferences represented by the cost functions $v_i(.)$ such that the cost of a set of chores $S$ is given by $v_i(S) = \sum_{c\in S} v_i(c)$, where $v_i(c)$ represents the cost of chore $c$ for agent $i$.

A natural and popular fairness notion in the context of indivisible items is called maximin share (MMS), introduced by~\citet{budish2011approxCEEI}. It appears to be also favored by participating agents in real-life experiments~\citep{GatesGD20}. An agent's MMS is defined as the minimum cost they can ensure by partitioning all the chores into $n$ bundles (one for each agent) and then receiving a bundle with the highest cost. Formally, for a set $S$ of chores and an integer $d$, let $\Pi_d(S)$ denote the set of all partitions of $S$ into $d$ bundles. Then, 

\[ \text{MMS}_i^d(S) := \min_{(S_1, \dots, S_d) \in \Pi_d(S)} \max_j v_i(S_j).\]

Let us denote the MMS  of an agent $i$ by $\vmms_i$ := MMS$_i^n(\items)$. In an MMS allocation, each agent's bundle cost does not exceed their MMS. However, for more than two agents, MMS allocations are not guaranteed to exist~\citep{Aziz2017,FeigeST21}. Therefore, the focus shifted to exploring approximations of MMS and identifying interesting special classes where MMS allocations can be achieved. Two natural relaxations are multiplicative and ordinal approximations. 

\paragraph{$\alpha$-MMS} This approach involves multiplying the MMS by a factor $\alpha > 1$ to raise each agent's threshold. An allocation is said to be $\alpha$-MMS if the cost of each agent's bundle is at most $\alpha$ times their MMS. Research has progressed to demonstrate the existence of $13/11$-MMS allocations~\citep{huang2023reduction}.

\paragraph{$1$-out-of-$d$-MMS} Another way to adjust the threshold is by considering the MMS when the chores are divided into $d<n$ bundles. An allocation is $1$-out-of-$d$-MMS if each agent's bundle cost is no more than MMS$_i^d(\items)$. This relaxation, initially introduced in~\citep{budish2011approxCEEI} for the case of goods, is valued for its resilience to small perturbations in chores costs; see~\citep{Hosseini2021OrdinalMS} for more details. The current best-known factor for which existence is established is 1-out-of-$\lfloor \frac{3}{4}n\rfloor$~\citep{HosseiniSS22}. 

\subsection{Our contributions}
In this paper, we advance the state-of-the-art on all three fronts: achieving exact MMS, and exploring both multiplicative and ordinal approximations. We establish the existence of 
\begin{itemize}
\item {\bf 1-out-of-$\lfloor \frac{9}{11}n\rfloor$ MMS allocations} for all additive instances, improving the current best-known factor of 1-out-of-$\lfloor \frac{3}{4}n\rfloor$. 
\item {\bf MMS allocations for factored instances}, where each $v_i(c) \in \{p_1, p_2, \dots, p_k\}$ such that $p_{\ell} = p_{\ell-1}\cdot q$ for some integer $q>0$ for each $\ell\in[k-1]$. Factored instances encompass the well-studied class of weakly lexicographic preferences~\citep{AzizBLLM19,HosseiniSVX21}. This contribution also resolves an open question posed by~\cite{ebadian2021fairly}.  
\item {\bf $15/13$-MMS allocations for personalized bivalued instances}, where each $v_i(c) \in \{a_i, b_i\}$ for some positive rational numbers $a_i, b_i$.
They generalize the well-studied bivalued instances, where each $v_i(c) \in \{a, b\}$ for some positive constants $a,b$, as the values of $a_i,b_i$ can vary between different agents. This is better than the factor of $13/11$ for general instances. 
\end{itemize}

We achieve these results by leveraging the Heterogeneous First Fit Decreasing (HFFD) algorithm of~\citet{HuangL21}. The HFFD algorithm is a heterogeneous variant of the classic First Fit Decreasing (FFD) algorithm used in the bin packing problem~\citep{Johnson1973}, with an approximation factor proven to be $13/11$~\citep{huang2023reduction}. As in previous works on MMS, the algorithm is straightforward; however, the novelty and challenge lie in the intricate analysis. 

For our first result, we extend the analysis of HFFD algorithm in~\citep{huang2023reduction} to the ordinal approximation of MMS, providing an improved bound. Our second result demonstrates that the HFFD algorithm is optimal for factored instances. Our third result establishes that the HFFD algorithm attains a factor of $15/13$ for personalized bivalued instances through a detailed case analysis. 

Additionally, our approach results in polynomial-time algorithms for the second and third contributions, enabling the computation of an MMS allocation for factored instances and a $15/13$-MMS allocation for personalized bivalued instances in polynomial time. 

\subsection{Further related work}
Given the intense study of MMS notion and its variants, we focus on closely related work here. 

Computing the MMS value of an agent is NP-hard, even for two agents, using a straightforward reduction from the Partition problem. However, a Polynomial Time Approximation Scheme (PTAS) exists for this computation using a PTAS for the job scheduling problem~\citep{HochbaumS87}. However, for factored instances, MMS values can be computed in polynomial time~\citep{ebadian2021fairly}.

MMS allocations are not guaranteed to exist for more than two agents~\citep{Aziz2017,FeigeST21}, which has motivated the exploration of approximate MMS allocations to ensure their existence. For multiplicative approximation, a series of works ~\citep{Aziz2017,barman2020approximation,HuangL21,huang2023reduction} have established the current best approximation factor of $13/11$. On the other hand, for ordinal approximation, research has progressed to show the existence of 1-out-of-$\lfloor \frac{3}{4}n\rfloor$ MMS allocations~\citep{Aigner-Horev2022,HosseiniSS22}. For the special case of (non-personalized) bivalued instances, MMS allocations are known to exist~\citep{Feigebivalued}.  

\paragraph{Goods} The MMS notion can similarly be defined for the fair division of goods. Like the case of chores, MMS allocations for goods do not always exist~\citep{KurokawaPW18}. Extensive research has been dedicated to approximate MMS allocations for goods. Notable works~\citep{amanatidis2017approximation,ghodsi2018fair,garg2021improved,simple,HosseiniS21,Hosseini2021OrdinalMS} have led to a multiplicative approximation factor of $3/4 + 3/3836$~\citep{AkramiG24} and an ordinal approximation factor of 1-out-of-$4\lfloor n/3 \rfloor$~\citep{AkramiGST24}.

\section{Preliminaries}
We denote $[t] = \{1, \dots, t\}$. 
An instance of the problem of dividing a set $\items$ of $m$ indivisible items (chores) among $n$ agents is given by a cost function vector $v = (v_1, \dots, v_n)$, where $v_i$ denotes the additive cost function of agent $i$. We assume that all chore costs are positive, i.e., $v_i(c) > 0$ for all $c\in\items$.

\paragraph{IDO instances} It is without loss of generality to assume that the instance has identical order preferences (IDO), meaning there exists a universal ordering of all chores $c_1, \dots, c_m$ such that, for any agent $i$, $v_i(c_1) \ge v_i(c_2) \dots \ge v_i(c_m)$~\citep{HuangL21}. Thus, we will assume a \emph{universal ordering} over the chores, which is also used for tie-breaking.

\paragraph{FFD algorithm} The chores division problem has a strong connection with the well-studied bin packing problem. \citet{HuangL21} were the first to apply the First Fit Decreasing (FFD) algorithm, a technique from the bin packing problem, to the fair division of chores, a method further utilized in~\citet{huang2023reduction}. 

Given a set $\items$ of chores, a cost function $v$ on $\items$, and a threshold $\tau$ (the "bin size"), the FFD algorithm first orders the chores by descending cost, based on the cost function $v$. It then packs each chore in the order into the first available bin where it fits, meaning the total cost in the bin does not exceed the threshold $\tau$. If no existing bin can accommodate the chore without surpassing the threshold, a new empty bin is opened, and the chore is placed in it. We denote the sequence of bins output by the FFD algorithm for a given set of inputs as $FFD(\items, v, \tau)$.  

\paragraph{MultiFit algorithm} To ensure that the number of bins exactly equals $n$ (one for each agent), we employ the MultiFit algorithm~\citep{coffman1978application}. This algorithm runs the FFD algorithm multiple times with different threshold values $\tau$. The goal is to find a $\tau$ that allows $FFD(\items, v, \tau)$ to allocate all the chores into exactly $n$ bins. It uses binary search to efficiently determine the threshold.  

\paragraph{HFFD algorithm} \citet{HuangL21} extended the FFD algorithm from bins to agents to accommodate the scenario where agents have distinct cost functions $v_i$'s and different thresholds $\tau_i$'s. This generalized algorithm is known as Heterogeneous FFD (HFFD). The HFFD algorithm processes chores by the universal ordering (from largest to smallest cost). It proceeds by filling one bin at a time as follows: for each chore, the algorithm checks if the chore fits into the current bin according to the threshold $\tau_i$ and the cost function $v_i$ of at least one remaining agent $i$. If it does fit, the chore is added to the bin; otherwise, it is skipped. Once no more chores fit into the current bin, the bin is closed and allocated to one of the remaining agents for whom the last chore fitted. This process continues until all chores are allocated or until no remaining agents can accommodate the chores. If all chores are successfully allocated, the algorithm is said to have \emph{succeeded}; if there are remaining chores that cannot be allocated to any agent, the algorithm \emph{failed}. 
\medskip

\noindent We next describe the concept of \emph{First-Fit-Valid}, as defined in~\citep{huang2023reduction}, which will be useful throughout the paper. We note that all the missing proofs are in the appendix. 

\subsection{First Fit Valid}
We first introduce notations to compare two bundles. For a given a bundle $B$, let $B[p]$ denote the $p$-th largest chore in $B$. If multiple chores share the same cost, we select $B[p]$ according to the universal ordering. We extend this notation to define $v(B[p])=0$ when $p > |B|$.

\begin{definition}[Lexicographic order]
Given a cost function $v$ and two bundles $B_1,B_2$, we say $B_1\lex{=} B_2$ if $v(B_1[p])=v(B_2[p])$ for all $p$; $B_1\lex{\le} B_2$ if $B_1\lex{=}B_2$  or  $v(B_1[q])<v(B_2[q])$, where $q$ is the smallest index for which $v(B_1[q])\neq v(B_2[q])$.

Note that every two bundles are comparable by this order.

We define $B_1\lex{\ge}B_2$ if $B_2\lex{\le}B_1$; and $B_1\lex{<}B_2$ if $B_1\lex{\le}B_2$ and $B_1\lex{\neq}B_2$; and similarly define $\lex{>}$. 
\end{definition}

Note that, when we consider only a single cost function, we can treat $\lex{=}$ as $=$. This is because, 
from the point of view of a single agent, chores with the same cost are equivalent, and can be exchanged anywhere, without affecting any of the arguments in the proofs.
 
Based on lexicographic order, we can define a \emph{lexicographically maximal} subset. 
\begin{definition}[Lexicographically maximal subset]
Given a set of chores $S$, a cost function $v$ and a threshold $\tau$, let $S_{\leq\tau}=\{B\subseteq S\mid v(B)\le \tau\}$ be the set of all bundles with cost at most $\tau$. A bundle $B$ is a \emph{lexicographically maximal subset} of $S$, w.r.t. $v$ and $\tau$,
if $B\in S_{\leq\tau}$ and $B\lex{\ge} B'$ for any $B'\in S_{\leq\tau}$.
\end{definition}

Based on lexicographical maximality, the following notion relates the outputs of the FFD and HFFD algorithms. 

\begin{definition}[Benchmark bundle]
Suppose we are given a set of chores $\items$, a collection of bundles $\allocs$, a cost function $v$,  and a threshold $\tau$. For each $k\in[n]$, we define the \emph{$k$-th benchmark bundle} of $\allocs$, denoted $B_k$, as a lexicographically maximal subset of $\items\setminus \bigcup_{j<k} \alloci[j]$ with respect to $v$ and $\tau$.
\end{definition}

Intuitively, once the bundles $\alloci[1],\ldots,\alloci[k-1]$ have been allocated, the benchmark bundle is the lexicographically maximal feasible subset of the remaining items. The following proposition shows that this is exactly the outcome of running the FFD algorithm for generating the next bundle $\alloci[k]$.

\begin{propositionrep}[\citealt{huang2023reduction}]
\label{prop-lexiffd}
Let $\mathbf{D}$ be the bundle collection which is output by FFD$\left(\items,v,\tau\right)$.
Then, for all $i\in[n]$, we have $D_i\lex{=}B_i$, 
where $B_i$ is the $i$-th benchmark bundle of $\mathbf{D}$.
\end{propositionrep}
\begin{proof}
By the description of FFD, we have $v(D_i)\le\tau$,
so $D_i\in S_{\leq \tau}$
where $S = \items\setminus \bigcup_{j<k} D_j$.
By the definition of lexicographically maximal bundle, we have $B_i\lex{\ge}D_i$. To show equality, it is sufficient to prove $D_i\lex{\ge}B_i$ for all $i\in[n]$.
Suppose the contrary, and 
let $i$ be the smallest integer for which $D_i\lex{<}B_i$.
Let $q$ be the smallest integer for which $v(D_i[q])<v(B_i[q])$. 
When FFD constructs bundle $D_i$,
the chore $B_i[q]$ is available, it is processed before $D_i[q]$,
and the cost of the current bundle would not exceed $\tau$ if $B_i[q]$ is put into it at that moment, since $v(B_i)\leq \tau$. 
Therefore, FFD would put $B_i[q]$ into the current bundle before $D_i[q]$ --- a contradiction.
So we have $D_i\lex{\ge}B_i$, which gives us $D_i\lex{=}B_i$ as claimed.
\end{proof}
For completeness, we provide a proof in the Appendix.

Next, we show that when using HFFD to generate the bundle in some iteration $k$, the resulting bundle is guaranteed to be lexicographically greater than or equal to the $k$-th benchmark bundle. This property is called \emph{First Fit Valid}, which is formally defined below.

\begin{definition}[First Fit Valid (FFV)]
\label{def-first-fit-valid}
Given a bundle collection $\allocs$, a cost function $v$ and a threshold $\tau$, we call the tuple $(\items, \allocs,v, \tau)$ \emph{First Fit Valid} if for any $k\in[n]$, we have $A_k\lex{\ge}\bundle_k$, where $\bundle_k$ is the $k$-th benchmark bundle of $\allocs$,\footnote{Note that $B_k$ is lexicographically-maximal among all subsets of $\items\setminus \bigcup_{j<k}A_j$ with cost at most $\tau$, but $A_k$ may be lexicographically larger than $B_k$ since its cost may be larger than $\tau$.} and the lexicographic comparison uses the cost function $v$.
\end{definition}
We emphasize that the definition of a tuple as FFV depends on the entire set of chores $\items$, not only on the chores allocated in $\allocs$.
This is because, if $\items$ contains chores not allocated in $\allocs$, this may affect the benchmark bundles.

With the notion of First Fit Valid, the next proposition characterizes the output of the HFFD algorithm, where $v_{\last},\tau_{\last}$ are the cost function and threshold of the last agent  $\last$ in HFFD.
\begin{propositionrep}[\citealt{huang2023reduction}]
\label{prop-HFFD-first-fit}
For any IDO instance $\inst$ and any threshold vector $(\tau_1,\ldots,\tau_n)$,
if Algorithm HFFD
produces a bundle collection $\allocs$,
then for any $\tau \leq \tau_{\last}$,
the tuple $(\items,\allocs,v_{\last},\tau)$ is a First Fit Valid  tuple. 
\end{propositionrep}
\begin{proof}
For all $k\geq 1$, 
we have to prove that $A_k\lex{\ge}B_k$, where $B_k$ is the $k$-th benchmark bundle of $\allocs$ w.r.t. threshold $\tau$, where the lexicographic comparison uses the cost function $v_{\last}$.

if $A_k\lex{=}B_k$ we are done. 
Otherwise, let $q$ be the smallest index such that $\vai[\last]{A_k[q]}\neq \vai[\last]{B_k[q]}$.
Suppose for contradiction that $\vai[\last]{A_k[q]}<\vai[\last]{B_k[q]}$.
This implies that $B_k[q]$ precedes $A_k[q]$ in the universal ordering. 
In addition, if we add chore $B_k[q]$ to the current bundle before chore $A_k[q]$ is added, the cost of the current bundle would be acceptable by agent $\last$,
since $v_{\last}(B_k)\leq \tau$ by definition of a benchmark bundle,
and $\tau\leq \tau_{\last}$ by assumption.
Therefore, HFFD would insert chore $B_k[q]$ to the current bundle before chore $A_k[q]$, which is a contradiction.

This implies that $\vai[\last]{A_k[q]}> \vai[\last]{B_k[q]}$.
Since $q$ is the first index in which $A_k$ and $B_k$ differ, 
$A_k\lex{>}B_k$.
\end{proof}

\section{ 1-out-of-$\lfloor\frac{9}{11}n\rfloor$ MMS}
In this section, we demonstrate that the HFFD algorithm yields a 1-out-of-$\lfloor\frac{9}{11}n\rfloor$ MMS allocation. Our approach combines the results of FFD for bin packing~\citep{DBLP:journals/tcs/DosaLHT13} with a reduction from HFFD to FFD, closely following the method described in~\citep{huang2023reduction}.

\begin{lemma}\label{lem:FFV-allocates-all-chores}
Let $n\geq 2$ be an integer, 
$\items$ a set of chores, 
and $v$ any cost function, 
Let $d := \lfloor\frac{9}{11}n\rfloor$, and suppose
\[ \text{MMS}^d(\items) = 1.\]
If the tuple $(\items, \allocs, v,1)$ is First Fit Valid, then the allocation $\allocs$ must contain all chores in $\items$.
\end{lemma}
\begin{proof}
We assume by contradiction that there is a chore $c^- \in \items$, that is not in any bundle in $\allocs$.

Let us consider the case where $v(c^-) \leq \frac{2}{11}$. Given that $\text{MMS}^d(\items) = 1$, the total cost of all chores in $\items$ is at most $n d$. This implies there exists at least one bundle $A_j$ whose cost is less than $d \leq \frac{9}{11}n$. Consequently, we have $v(A_j \cup \{c^-\})\leq 1$. However, $A_j \cup \{c^-\}$ is lexicographically larger than $A_j$, thus contradicting the definition of First Fit Valid.

We can now assume that $v(c^-) > \frac{2}{11}$. Because removing all chores less than $c^-$ would not affect the First Fit Valid. This operation is dicussed in Tidy-up Procedure in~\citep{huang2023reduction}.

We will construct another instance $\items'$ by decreasing the cost of some chores and removing some chores from $\items$.  Then we can obtain another allocation $\allocs'$, which is output of FFD$(\items', v, 1)$. The unallocated chore $c^-$ will still be in $\items'$, and not in the output of FFD. Notice that for this new instance the maximin share for $\lfloor\frac{9}{11}n\rfloor$ bins is no more than 1. Then we have a contradiction here. Because FFD algorithm can allocate all chores in this case \citep{DBLP:journals/tcs/DosaLHT13}. 

We now show how to construct such an instance, using notions from \citep{huang2023reduction}. 
    
We first define \emph{redundant chores}. Given a bundle $A_j$, a chore $A_j[p]$ is called \emph{redundant} if $p>k$ where $k$ is the smallest integer such that $v(\cup_{p\le k}A_j[p])\ge1$.  We remove all redundant chores from $\items$ and $\allocs$ at the beginning. We keep the notation $\items$ and $\allocs$, but assume that there is no redundant chores in the rest of the proof.

To see why they are called redundant, we can go back to the algorithm. If we set threshold 1 for FFD, then FFD will not allocate to any bundle chores with total cost more than 1.  So removing redundant chores does not influence the fact that there is an unallocated chore. 

Let $k$ be the smallest integer such that $v(A_k)>1$. Apply Proposition \ref{prop-lexiffd} to the first $k-1$ bundles of $\allocs$, they are the same as the output of FFD$(\items,v,1)$. To fix the $k$-th bundle, we introduce the following concept from~\citep{huang2023reduction}.
 
\begin{center}
\noindent\fbox{\parbox{0.95\textwidth}{
\begin{definition}[Suitable reduced chore]\label{def-suitable-redu} 
Let $(\items, \allocs, v,1)$ be a First Fit Valid tuple without any redundant chores. Let $k$ be the smallest integer such that  $v(\alloci[k])>1$. Let $c_k^*$ be the smallest chore in $A_k$.
    Let $c_k^{\circ}$ be a chore with $v(c_k^{\circ})<v(c_k^*)$ and $\items'=\items\setminus\{c^*_k\}\cup\{c_k^{\circ}\}$. 
    
Let $\mathcal{P}$ be the output of $FFD(\items',v,1)$. The  chore $c_k^{\circ}$ is a \emph{suitable reduced chore} of $c_k^*$ if
  \begin{equation}
  \label{eq:suitable-redu}
     \bigcup_{j\le k}P_j=\left(\bigcup_{j\le k}A_j\right)\setminus\{c_k^*\}\cup\{c_k^{\circ}\}.
  \end{equation}

In other words: replacing $c_k^{*}$ by $c_k^{\circ}$ has no effect on the allocation of bundles $i>k$.
\end{definition}
}}
\end{center}

Note that we change Tidy-up tuple with First Fit Valid tuple in the above definition. If we can always find a suitable reduced chore, then we can fix the allocation to the output of FFD algorithm one bundle by one bundle. The remaining thing is to prove there is a suitable reduced chore. 

Roughly, when we want to find a suitable reduced chore for bundle $A_k$, its reduced cost should be such that the new cost of $A_k$ is at most $1$. This puts an upper bound on the cost of suitable reduced chores. This motivates the following definition.

\begin{definition}[Fit-in space~\citep{huang2023reduction}]
Given an index $k$,
let $c^*_k$ be the last (smallest) chore in bundle $\alloci[k]$. 
The \emph{fit-in space} for $\alloci[k]$ is $ft(k) := 1-v(\alloci[k]\setminus\{c^*_k\})$.
\end{definition}

Note that since there are no redundant chores, the fit-in space is non-negative. The following claim is useful for simplifying our proof. 
\begin{claim}[\citealt{huang2023reduction} (Lemmas 11 and 12)]\label{claim:small-fit-in}
    Let $\gamma=\frac{2}{11}$. If there are no redundant chores, and no chore is smaller than the unallocated chore $c^-$, and $ft(k) \le 2\gamma=\frac{4}{11}$, then a suitable reduced chore exists. 
\end{claim}

Example \ref{example:reduction} after the proof gives some details that how this reduction work. It can help understand the whole proof.

Now we only need to consider the case of a large fit-in space $ft(k)> 2\gamma=\frac{4}{11}$. 
This is the main point in which our proof differs from theirs. 
In their proof, they used a procedure called ``Tidy-Up'', which makes the instance easier to work with. We cannot use this procedure here, as it might reduce the number of bundles. Therefore, we need a whole different proof for the case of a large fit-in space.

Let $c_k^*$ be the smallest chore in $A_k$.
We prove that, for the case of large fit-in space, there are exactly 2 chores in bundle $\alloci[k]$.

No chore could be larger than 1, as this is the maximin share for $\lfloor\frac{9}{11}n\rfloor$ bins. 
On the other hand, there cannot be three chores in bundle $A_k$. Otherwise, it means that there are at least two chores except $c_k^*$. 
When the cost $v(A_k)>1$, the smallest chore is larger than  the fit-in space by definition. So the cost of each chore in $A_k$ is at least $\frac{4}{11}$. So the total cost of $A_k\setminus\{c_k^*\}$ is at least $2\cdot\frac{4}{11}=\frac{8}{11}$. By the definition of fit-in space, we have $v(A_k\setminus\{c_k^*\})=1-ft(k)<\frac{7}{11}$.  We have a contradiction here. Therefore, bundle $A_k$ contains exactly 2 chores, namely $A_k[1]$ and $A_k[2] = c_k^*$.

Let $c_k^{\circ}$ be a chore with cost $v(c_k^{\circ})=ft(k)$. We  prove it is a suitable reduced chore. 
Let us run FFD for the first $k$ bundles on the instance after replacing chore $c_k^*$ with $c_k^{\circ}$. 

Suppose first that $c_k^{\circ}$ is not allocated to a bundle $P_i$ with $i<k$. By selection of $k$, this means that the first $k-1$ bundles generated by FFD are equal to those of $\allocs$ ($P_i=A_i$ for all $i<K$), then when filling $P_k$, the chores $A_k[1]$ and $c_k^{\circ}$ are available.  Because $\allocs$ is FFV, $A_k[1]$ must be the largest remaining chore, so it will be allocated to $P_k$ first.  Now, the remaining room in $P_k$ is exactly equal to the fit-in space, which is the size of $c_k^{\circ}$.  Therefore, $c_k^{\circ}$ will be allocated to $P_k$, and we are done.

Now we consider the case that $c_k^{\circ}$ is allocated to a bundle $P_i$ where $i<k$.  We have the following claim.

\begin{claim}\label{claim:equal-chores}
Suppose that $c_k^{\circ}$ is allocated to a bundle $P_i$ where $i<k$, we have $v(P_i[1])=v(P_k[1])=v(A_k[1])=v(P_j[1])$ for all $i<j<k$.
\end{claim}

\begin{claimproof}
We first show  $v(A_i[1])=v(P_i[1])$ when $i \le k$. As allocation $\allocs$ is First Fit Valid w.r.t. $\items$, and the costs of first $k-1$

As FFD processes large chores before small chores, any chore larger than $c_k^*$ that is allocated in $A_j$ where $j<k$, will be allocated in $P_j$, as its allocation is not influenced by $c_k^*$.

We also have $v(A_k[1])=v(P_k[1])$ by FFV, as it is the largest chore remaining after removing the chores in the first $k-1$ bundles.

By the properties of FFD, we have $v(P_i[1])\ge v(P_k[1])=v(A_k[1]) = 1-ft(k)$, as $A_k[1]$ is the only chore in $A_k$ except the last one.

If $c_k^{\circ}$ is allocated to $P_i$ for $i<k$, we have $v(P_i[1]\cup c_k^{\circ})\le 1$, which means $v(P_i[1])\le 1-v(c_k^{\circ}) = 1-ft(k)$. 

Combining the two inequalities on $v(P_i[1])$ leads to $v(P_i[1])=1-ft(k)=v(A_k[1])$. As $v(P_k[1])\le v(P_j[1])\le v(P_i[1])$ for $i<j<k$, the costs of all these chores must be equal.
\end{claimproof}

With Claim~\ref{claim:equal-chores}, we have an easy description of what happened between $\mathcal{P}$ and $\mathcal{A}$. We prove that there will be a shift: for all $i< j\le k$, we have $P_j=\{A_j[1]\}\cup (A_{j-1}\setminus \{A_{j-1}[1]\})$. 
First, we have $P_i=\{A_i[1]\}\cup \{c_k^{\circ}\}$. Basically, the argument is as follows: If you give the same space and same set of chores, the FFD algorithm will get the same result.  We prove that it is true for $i+1$. The similar argument directly works for other indexes. Let $S$ be the set of chores $A_i\setminus\{A_i[1]\}$. The set of chores $S$ are allocated into a space $1-v(A_i[1])=ft(k)$. After allocating the chore $P_{i+1}[1]$, the remaining space is $ft(k)$.  When FFD algorithm try to generate $P_{i+1}$, let us consider the set of chores which are no larger than $ft(k)$. Let us denote the set by $L=\{c\mid \items'\setminus(\cup_{j\le i}P_i)\}$. By Proposition \ref{prop-lexiffd}, the first $i$ bundles of allocation $\allocs$ are the same to the output of FFD$(\items,v,1)$. Let us consider the set of chores which are no larger than $ft(k)$ when $A_i$ is generated. Let us denote the set by $R=\{c\mid v(c)\le ft(k)\text{ and } c\in \items\setminus(\cup_{j<i} A_j)\}$. We have $R=L$. So what are allocated to $A_i$ will be also allocated to $P_{i+1}$.
For other bundles, the argument is similar. 

The chores in the bundle with index later than $k$ will not be affected the allocation of first $k$ bundle. So $c_k^{\circ}$ is a suitable reduced chore. 
\end{proof}
\begin{example}\label{example:reduction}
    This example illustrates the reduction from a FFV instance to a valid output for the FFD algorithm. Let us assume that the FFD threshold is set at $1$. There are four bundles with chore costs as follows: $A_1=\{.70,.25,.25\}, A_2=\{.60,.35\}, A_3=\{.60,.50\}, A_4=\{.54,.24,.23\}$. There is also one unallocated chore with a cost of $.19$.  

    We can see that bundles $A_1,A_3,A_4$ are invalid for FFD algorithm, as their total cost is larger than 1. The reduction will fix them one by one. 
    \begin{itemize}
        \item For bundle $A_1$ the fit-in space is $ft(1)=.05$, which is less than $\gamma=\frac{2}{11}$, 
        Therefore, we just remove the last chore (with cost $.25$).
        \item For bundle $A_3$, the fit-in space is $ft(3)=.40$, which is greater than $2\gamma=\frac{4}{11}$. Here, reducing the cost of $.50$ to $.40$ is a suitable reduced chore. The suitable reduced chore will be allocated to $A_2$ and kick out the last chore in bundle $A_2$. The last chore in $A_2$ will be allocated in bundle $A_3$.
        \item For the bundle $A_4$, the fit-in space is $ft(4)=.22$. By the result from \cite{huang2023reduction}, we can always find a suitable reduced chore here. In this particular case, any cost between $.19$ and $.22$ is Ok. Let us choose $.20$ for the suitable reduced chore. 
    \end{itemize}
    After all these operations, we get the following bundles. We have $A_1'=\{.70, .25\}, A_2' = \{.60, .40\}, A_3'=\{.60, .35\}, A_4'=\{.54, .24, .20\} $. The unallocated chore with cost $.19$ is still there. The resulting allocation is the same as the output of the FFD algorithm on the set of chores.
\end{example}
\begin{theorem}
Suppose that the MMS value for $\lfloor\frac{9}{11}n\rfloor$ bundles (bins) is 1 for every agent.  If the HFFD  algorithm is executed with a threshold of at least $1$ for each agent, then all chores are allocated.
\end{theorem}
\begin{proof}
Let $\allocs$ be the output of HFFD. 
By Proposition \ref{prop-HFFD-first-fit}, 
since $\tau_{\last}\geq 1$, the tuple 
$(\items, \allocs, v_{\last}, 1)$ is First-Fit-Valid. 

Given that the 1-out-of-$\lfloor \frac{9}{11} n\rfloor$ MMS of $v_{\last}$ is $1$ by assumption, 
we can apply Lemma \ref{lem:FFV-allocates-all-chores}
with $v = v_{\last}$,
and thus conclude that $\allocs$ must contain all chores.
\end{proof}

\section{Reduction: Factored \& Bivalued Instances}
\label{sec:reduction}
We will demonstrate that for factored and personalized-bivalued instances, the approximation ratio of HFFD is equal to that of MultiFit.

To facilitate our proof, we will use a swapping operation extensively. The swapping process is defined as follows: 
\begin{definition}
Given an allocation $\allocs$, a set of items $T_i\subseteq A_i$ and $T_j\subseteq A_j$, let allocation $\allocs'$ be the allocation after swapping $T_i$ and $T_j$. Formally, we have $A'_k=A_k$ if $k\neq \{i,j\}$  and $A'_i=(A_i\setminus T_i)\cup T_j$ and $A'_j=(A_i\setminus T_j)\cup T_i$. We define $$\swap{\allocs}{T_i}{T_j}=\allocs' \enspace .$$
\end{definition}

When the subscripts $i$ and $j$ are not crucial or their meaning is clear from the context, we will use symbols without subscripts. Additionally, there may be cases where one set of chores could be empty.
For example, the operation $\swap{\mathcal{A}}{\emptyset_j}{T}$ denotes moving the set of chores $T$ to bundle $A_j$. The subscript $j$ specifies the bundle receiving the chores in $T$.

\begin{lemma}[Reduction Lemma]
\label{lem:reduction}
For any threshold $\tau>0$, and any  cost function $v$
that is either factored or bivalued if FFD($\items, v,\tau$) allocates all the chores into $n$ bins,
and $(\items, \mathcal{Q},v, \tau)$ is a First-Fit Valid tuple,
then $\mathcal{Q}$ contains all chores in $\items$.
\end{lemma}
We will prove the above lemma for factored cost functions in Section \ref{sec:factored} and for bivalued cost functions in Section \ref{sec:bivalued}.

This lemma implies two things: 

\begin{corollary}
\label{cor:hffd}
    For both factored and personalized-bivalued instances, 
    we can reduce multiple agents case to single agent case. That is, the approximation ratio of HFFD is equal to that of MultiFit.
\end{corollary}
\begin{proof}
\label{cor:monotonicity}
This follows from Lemma \ref{lem:reduction} as the outcome of HFFD when all agents have threshold $\tau$ is First Fit Valid for $v_{\last}$ and $\tau$.
\end{proof}

\begin{corollary}
\label{cor:monotonicity}
For both factored and personalized-bivalued instances, FFD is \emph{monotone} in the following sense:
if FFD$(\items,v,\tau)$ allocates all chores, then FFD$(\items,v,\beta)$ allocates all chores for any $\beta>\tau$. 
\end{corollary}
\begin{proof}
This follows from Lemma \ref{lem:reduction} as the outcome of FFD$(\items,v,\beta)$ is First Fit Valid for $v$ and $\tau$.
\end{proof}

\paragraph{Polynomial-time computation} Corollary \ref{cor:monotonicity} implies that we can use binary search to design polynomial-time algorithms for factored and personalized bivalued instances. The process involves performing a binary search to determine a minimal threshold $\tau_i$ for each agent $i$, such that FFD$(\items,v_i,\tau_i)$ can allocate all chores. These thresholds $\tau_i$ are then used in the HFFD algorithm. By Corollary \ref{cor:monotonicity}, the HFFD algorithm will allocate all the chores, ensuring that each agent $i$ gets a bundle no larger than $\tau_i$.

However, this approach does not extend to the general case, as discussed in~\citep{huang2023reduction}. In general, combining all thresholds into the HFFD algorithm may not guarantee the allocation of all chores.

\section{Factored Instance}
\label{sec:factored}
A cost function is called \emph{factored} if any smaller cost is a factor of the next-larger cost, i.e., 
if $v(c_m) \leq \cdots \leq v(c_1)$ then
$v(c_m)\mid v(c_{m-1})\mid \dots\mid v(c_2)\mid v(c_1)$.
The notion first appeared with relation to bin packing in \citep{coffman1987bin}, where it is called ``divisible item sizes''. It was introduced into fair division in \citep{ebadian2021fairly}.

\begin{claim}[\cite{coffman1987bin}]
\label{Claim:exist-subset}
For a factored instance, given an index $j$ and a set $S$ of indices such that 

(1) $v(c_j)\ge v(c_i)$ for any $i\in S$ and 

(2) $ v(S)\ge v(c_j)$, 

\noindent
there is a subset $S'\subseteq S$ such that $v(S')=v(c_j)$.
\end{claim}

\begin{theorem}[\cite{coffman1987bin}]
\label{thm:ffd}
    In job scheduling with a factored cost function, FFD with threshold exactly equal to the minimum makespan always allocates all jobs. 
\end{theorem}

We now extend Theorem \ref{thm:ffd} from FFD to HFFD. The following is a special case of Lemma \ref{lem:reduction} for factored instances.

\begin{lemma*}[Reduction Lemma for Factored Instances]
For any threshold $\tau>0$, and any factored cost function $v$, if FFD($\items, v,\tau$) allocates all the chores into $n$ bins,
and $(\items, \mathcal{Q},v, \tau)$ is a First-Fit Valid tuple,
then $\mathcal{Q}$ contains all chores in $\items$.
\end{lemma*}
\begin{proof}
Here is a high-level idea for the proof. Given an allocation $\mathcal{P}$ produced by the FFD algorithm and a First-Fit Valid allocation $\mathcal{Q}$, our goal is to transform $\mathcal{P}$ into $\mathcal{Q}$ using a series of swapping operations. 

We describe the swapping operations in Algorithm \ref{alg:reduce-factored} and provide an illustrative example below. 

\begin{example}
    Suppose the threshold is $10$ and $P_1 = \{a,b,c,d,e\}$ with chore costs $4,2,2,1,1$, and $Q_1 = \{w,x,y,z\}$ with chore costs $4,4,4,1$.
    Note that the instance is factored and $Q_1\lex{>}P_1$.
    At the inner for loop, the condition $v(Q_k[j])>v(P_k[j])$ is first satisfied for $j=2$.
    The set $T = \{b,c,d,e\}$ and its total cost is $6$ which is larger than $v(Q_k[2])$.
    We find the set $S = \{b,c\}$, remove it from $P_1$ and swap it with $\{x\}$.
    Now we have $P_1 = \{a,x,d,e\}$ with chore costs $4,4,1,1$.

        Next, after the swap, $v(Q_k[j])>v(P_k[j])$ for $j=3$. We have $T = \{d,e\}$, its total cost is $2$ which is now smaller than $v(Q_k[j])$. So we remove $T$ from $P_1$ and swap it with $\{y\}$. 
    The new $P_1 = \{a,x,y\}$ with chore costs $4,4,4$.

    Next, for $j=4$ we have 
    $v(Q_k[j])>v(P_k[j])$ as $v(P_k[j])=0$ by definition. $T=\emptyset$, and we swap it with $z$. So the final $P_1$ is $\{a,x,y,z\}$, which is lex-equivalent to $Q_1$.
\end{example}

\begin{algorithm}[t]
\KwData{$\mathcal{Q}$ a First-Fit Valid allocation; $\mathcal{P}$ output of FFD on same items.}
\For{k=1 \emph{\KwTo} N}{
    \If{$Q_k \lex{=} P_k$}{Continue to next $k$.} 
    \tcc{Else, $Q_k \lex{>} P_k$ by FFV}
    \tcc{Compare $Q_k$ with $P_k$ from largest to smallest chore}
    \For{$j$= 1 \emph{\KwTo} $|Q_k|$}    {
        \If{ $v(Q_k[j])>v(P_k[j])$}{
        Let $T=\cup_{w\ge j}P_k[w]$\; 
        \If{$v(T)\ge v(Q_k[j])$}{Find a set $S\subseteq T$ such that $v(S)=v(Q_k[j])$\;\tcc{Claim 
        \ref{Claim:exist-subset} guarantees existence of $S$}
        Update $\mathcal{P}=\swap{\mathcal{P}}{S}{Q_k[j]}$
        \label{line:swap-s}
        }\Else{
        Update $\mathcal{P}=\swap{\mathcal{P}}{T}{Q_k[j]}$
        \label{line:swap-t}
        }}
    }
}
\caption{Reduction by swapping for a factored cost function}\label{alg:reduce-factored}
\end{algorithm}

\begin{claim}\label{claim:factored}
    At each iteration $k$ of the main loop, the cost of any bundle $P_i$ with $i>k$ does not increase.
\end{claim}
\begin{proof}
The proof is by induction. The base $k=0$ is straightforward. For $k\geq 1$, assume that the claim holds for all iterations before $k$. We have $v(P_k)\le\tau$ since $\tau$ is the threshold used by FFD. By the First-Fit Valid property, if we check the chores of $P_k$ and $Q_k$ from largest to smallest, the first different place $j$ must satisfy $v(Q_k[j])>v(P_k[j])$. 

In the swap of line \ref{line:swap-s}, the swap is between two sets of equal cost, $v(S) = v(Q_k[j])$, so the costs of all bundles to not change. 

In the swap of line \ref{line:swap-t}, the swap is between sets with different costs, $v(T) < v(Q_k[j])$,
so the cost of $P_k$ increases and the cost of some bundle $P_i$ with $i>k$  decreases.
\end{proof}

By Claim \ref{claim:factored}, at the start of each iteration $k$, $v(P_k)\leq \tau$, so by First Fit Valid definition, it still holds that $Q_k \lex{\ge} P_k$.

During iteration $k$, after each iteration $j$ of the inner loop, the set $P_k[1,\ldots,j]$ is leximin-equivalent to $Q_k[1,\ldots,j]$.
If, for some $j\leq |Q_k|$, we have $|P_k|<j$, then $T=\emptyset$, and the algorithm will move some chores from some later bundles into $P_k$.
Thus, at the end of iteration $k$, we will have $P_k \lex{=} Q_k$. Consequently, the entire process ensures that $\mathcal{P}\lex{=}\mathcal{Q}$. 
As allocation $\mathcal{P}$ contains all chores, allocation $\mathcal{Q}$ must also include all chores. 
This completes the proof of Lemma~\ref{lem:reduction} for factored instances.
\end{proof}

\begin{theorem}
In chore allocation with factored cost functions,
HFFD with the thresholds of all agents equal to their maximin share always allocates all chores.
\end{theorem}

\begin{proof}
This follows from Lemma \ref{thm:ffd} and Corollary \ref{cor:hffd}.
\end{proof}

\section{Bivalued Instance}
\label{sec:bivalued}
In this section, we prove Lemma \ref{lem:reduction} for bivalued instances.
Most proofs of the auxiliary claims are deferred to the Appendix.

\begin{lemma*}[Reduction Lemma for Bivalued Instances]
For any threshold $\tau>0$, and any bivalued cost function $v$, if FFD($\items, v,\tau$) allocates all the chores into $n$ bins,
and $(\items, \mathcal{Q},v, \tau)$ is a First-Fit Valid tuple,
then $\mathcal{Q}$ contains all chores in $\items$.
\end{lemma*}

The reduction process for this case is described in Algorithm \ref{alg:reduce-bivalue}, and we provide an illustrative example below. For any bundle $A$, we denote by $L(A)$ the set of large chores in $A$, and by $S(A)$ the set of small chores in $A$.

\begin{algorithm}[t]
\KwData{$\mathcal{Q}$ a First-Fit Valid allocation; $\mathcal{P}$ output of FFD on same items.}
\For{k=1 \emph{\KwTo} N}{
    \If{$Q_k \lex{=} P_k$}{Continue to next $k$.}
    \tcc{Else $Q_k \lex{>} P_k$ by FFV, so either {\small $|L(Q_k)|>|L(P_k)|$}, or {\small $|L(Q_k)|=|L(P_k)|$} and {\small $|S(Q_k)|>|S(P_k)|$}.}
    \If{$|L(Q_k)|>|L(P_k)|$}
    {Let $cl$ be the large chore which appears last in allocation $\mathcal{P}$\;
    \tcc{We prove later that chore $cl$ must be in some bundle $P_i$ for $i>k$}
    Update $\mathcal{P}=\swap{\mathcal{P}}{cl}{S(P_k)}$\;
    }
    \tcc{At this point we can assume wlog that $Q_k\supseteq P_k$.}
    \For{$c\in Q_k\setminus P_k$}
    {
    Let chore $cl$ be the one that appears last in allocation $\mathcal{P}$ with the same cost as $c$\;
    Update $\mathcal{P}=\swap{\mathcal{P}}{cl}{\emptyset_k}$\;
    }
    }
\caption{Reduction by swapping for a bivalued cost function}\label{alg:reduce-bivalue}
\end{algorithm}

\begin{example}

    Suppose the threshold is $20$ and $P_1 = \{a,b,c,d,e\}$ with chore costs $7,7,2,2,2$
    and $Q_1 = \{u,v,w,x,y,z\}$ with chore costs $7,7,7,7,2,2$.
    We have $|L(Q_1)| = 4 > 2 = L(P_1)$,
    so there must be a large chore in some bundle $P_i$ with $i>1$; we take the last such chore (say $x$) and swap it with $c,d,e$, so now $P_1 = \{a,b,x\}$ with chore costs $7,7,7$.

    Now, we can assume w.l.o.g. that the large chores in $P_1$ are the same as some large chores in $Q_1$, say $u,v,x$, so $P_1 = \{u,v,x\}$ and $P_1\subseteq Q_1$.
    In the next loop, we simply add to $P_1$ the missing chores $w,y,z$ or some other chores with the same costs.
    
\end{example}
\begin{claim}
1. In iteration $k$, the cost of any bundle $P_i$ with  $i>k$ does not increase. 

2. If initially the number of large chores in $Q_k$ is greater than in $P_k$, then the number of large chores in $P_k$ does not change in iterations $1$ to $k-1$. 
\end{claim}
\begin{proof}
We prove both claims by induction. 

For $k=0$, both claims hold trivially. 
Suppose that both claims are true before iteration $k$. We prove that both claims hold after iteration $k$.

Suppose first that $|L(Q_k)|>|L(P_k)|$.
At the start of iteration $k$, the algorithm has already made all bundles $1,\ldots,k-1$ equal in $\mathcal{P}$ and $\mathcal{Q}$.
Since $\mathcal{P}$ allocates all large chores, there must be another bundle $i>k$ such that 
$|L(Q_i)|<|L(P_i)|$. In particular, $P_i$ contains at least one large chore, so the last large chore in $\mathcal{P}$ is not in bundle $P_k$.
Therefore, in iterations $1,\ldots,k-1$, no large chore was taken from $P_k$.
The algorithm also never adds large chores to later bundles. Overall, the number of large chores in $P_k$ does not change.

As $P_k$ is generated by FFD algorithm, 
and at least one large chore is processed later than $P_k$, 
adding one more large chore to $P_k$ would exceed the threshold. Therefore, the cost of $S(P_k)$ is smaller than one large chore. 
So when we do the swapping, the cost of bundles with index larger than $k$ does not increase.
This means that, before iteration $k$, we have $v(P_k)\le\tau$. By the First-Fit Valid property, we have $|L(Q_k)|\ge |L(P_k)|$, so the algorithm  is valid.
\end{proof}

As in the factored case, at the end of iteration $k$, we have $P_k \lex{=} Q_k$, so at the end of the algorithm  $\mathcal{P}\lex{=}\mathcal{Q}$. 
As $\mathcal{P}$ contains all chores, $\mathcal{Q}$ contains all chores too. 
This completes the proof of Lemma \ref{lem:reduction} for bivalued instances.

Following Corollary \ref{cor:hffd},
in order to prove an approximation ratio of HFFD on bivalued instances, it suffices to prove an approximation ratio of MultiFit. We are not aware of previous work analyzing the performance of MultiFit on bivalued instances. In this section, we establish a tight approximation ratio of $15/13$.

Throughout this section, we denote the large chore size by $l$, the small chore size by $s$, and the MMS value by $\vmms$.

\subsection{Lower bound}
The lower bound of $15/13$ is shown using a simple example.
\begin{example}
\label{exm:15/13}
There are $n=3$ agents. All agents have the same valuation.
There are three large chores with cost $l=4$,
and nine small chores with cost $s=3$. 
~
The MMS is $13$, as the chores can be partitioned into three bundles containing 1 large chore and 3 small chores each.
~
However, FFD with any bin size in $[13,15)$ packs the three large chores into a single bin, and eight small chores into the following two bins, so one small chore remains unallocated.
\end{example}

\subsection{Upper bound}
Our proof for the upper bound uses a technique similar to that of Section \ref{sec:factored}. We establish the following theorem using an involved case analysis.

\begin{theorem}
\label{thm:bivalued}
    For personalized bivalued instances of chores, the HFFD algorithm achieves a tight approximation ratio of $\frac{15}{13}$ for MMS.
\end{theorem}

\begin{proofsketch}
    Algorithm~\ref{alg:simulate-bivalue} 
    accepts as input an allocation $\mathcal{Q}$ with maximum cost $\vmms$,
    and an allocation $\mathcal{P}$ generated by FFD with threshold $\tau := \frac{15}{13}\vmms$. 
    The algorithm processes both allocations bin by bin. In each iteration $k$, it modifies $\mathcal{Q}$ such that $Q_k$ becomes equal to $P_k$.
    The modification is done by swapping some chores in $Q_k$ with some chores from bundles $Q_z$ for $z>k$.
    Using a detailed case analysis we show that, even after the swap, the cost of $Q_z$ remains at most $\tau$.
    At the end of iteration $n$,
    $Q_i=P_i$ for all $i\in[n]$.
    As $\mathcal{Q}$ contains all chores, this implies that 
    $\mathcal{P}$ allocates all chores in $n$ bins, which concludes the proof.
\end{proofsketch}

We now provide the complete proof.

We are given a single cost function $v$,
and a list $\mathcal{Q}$ of $n$ bundles such that the value of each bundle is at most $\vmms$ ($\mathcal{Q}$ can represent the MMS partition of an agent with cost function $v$ and MMS value $\vmms$).

We set a threshold $\tau := \frac{15}{13}\vmms$, and prove that FFD with bin-size $\tau$ packs all the items in $\mathcal{Q}$ in at most $n$ bins.

We first handle the case that $s$ is small w.r.t. $\vmms$.

\begin{lemma}
\label{lem:tau=H+s}
    For any $l$ and $s$, FFD with bin size $\tau \geq \vmms+s$ packs all chores in at most $n$ bins.
\end{lemma}
\begin{proof}
    FFD processes the large chores first. 
    Every bin (except maybe the last one) receives at least $\lfloor \vmms/l\rfloor$ large chores, whereas every bin in the optimal packing receives at most $\lfloor \vmms/l \rfloor$ large chores. Hence, all large chores are packed into at most $n$ bins.
    
    FFD then processes the small chores. As the total cost of all chores is at most $n\cdot \vmms$, at least one of the first $n$ bins has a total cost of at most $\vmms$; this bin has at least $\tau-\vmms\geq s$ remaining space. Therefore, every small chore can be packed into one of the first $n$ bins.
\end{proof}

\begin{corollary}
\label{cor:6.5s}
 if $\vmms\geq 6.5s = \frac{13}{2}s$, then FFD with bin size $\tau=\frac{15}{13}\vmms$ packs all items into $n$ bins. 
\end{corollary}
Based on the corollary, from now on we focus on the case $\vmms<6.5s$.

\begin{algorithm}
    \KwData{
	    $v$ --- a cost function;
	    \\
	    $\mathcal{Q}$ --- an $n$-maximin-share partition of $v$, with MMS value $\vmms$;
	    \\
	    $\tau$ --- a threshold for FFD, $\tau := \frac{15}{13}\cdot \vmms$;
	    \\
	    $\mathcal{P}$ --- the output of FFD with cost function $v$
	    and bin size $\tau$.
    }
    \noindent\rule{8cm}{0.4pt}
    \\

    Order the bundles in $\mathcal{Q}$ by the number of large chores, from many to few. Break ties by the number of small chores. \\
    \label{step:order} 
    
    \For{k=1 \emph{\KwTo} n}{
         {\small // \emph{INVARIANT 1: $Q_i = P_i$ for all $i < k$.}
         \\
         // \emph{INVARIANT 2: $v(Q_k)\leq \vmms$, and $v(Q_i) \leq \tau$ for all $i > k$.}
         \\
         // \emph{INVARIANT 3: one of the following holds for all $i> k$:
         (a) $v(Q_i) \leq \vmms$, or 
         (b) $Q_i$ contains no large chores, or 
         (c) $Q_i$ contains one large chore and some $b_i$ small chores and $(b_i+2)s\leq \tau$.
         }}\\
    \If{$|Q_k|>|P_k|$}     {
       \label{step:inner-loop-1}
	    Let $S$ be a set of two small chores in $Q_k$.\\
	    {\small \emph{~~~~// Lemma \ref{lem:twosmallchores} shows that they exist.}}
	     \label{step:twosmallchores}
	
	    Let $z := $ max index of a bundle in $\mathcal{Q}$ that contains a large chore.\\
	    {\small \emph{~~~~// Lemma \ref{lem:twosmallchores} shows that $z>k$.}}
	    
	    Let $cl$ be a large chore in $Q_z$\;
	    
	    Update $\mathcal{Q}=\swap{\mathcal{Q}}{cl}{S}$\;\label{step:swap1}
		{\small \emph{// In two special cases, another swap is needed:}}\\    
	    \If{$|Q_k|>|P_k|$} { 
		\label{step:inner-condition-2}	    
	    	\If {$Q_k$ contains exactly 2 large and 1 small chore, and $P_k$ contains exactly 2 large chores} {
	    		Move one small chore from $Q_k$ to $Q_z$.
	    	}
	    	\ElseIf{$Q_k$ contains exactly 2 large and 2 small chores, and $P_k$ contains exactly 3 large chores} {
			    Let $S'$ be a set of two small chores in $Q_k$. 
			    
			    Let $z' := $ max index of a bundle in $\mathcal{Q}$ that contains a large chore.

			    Let $cl'$ be a large chore in $Q_{z'}$\;
			    
			    Update $\mathcal{Q}=\swap{\mathcal{Q}}{cl'}{S'}$\;				
	    	}
	    	
	    }
    } 
            
    // {\small \emph{Here, $|Q_k|\leq |P_k|$.}}
    
    \For{j=1 \emph{\KwTo} $|P_k|$}  {\label{step:inner-loop-2}     
        \If{$v(Q_k[j])<v(P_k[j])$} { 

	    	Let $z := $ max index of a bundle in $\mathcal{Q}$ that contains a chore with the same cost as  $P_k[j]$, and let $cl$ be such a chore in $Q_z$\;
                {\small \emph{// either $P_k[j]$ is large and $Q_k[j]$ is small/empty, or $P_k[j]$ is small and $Q_k[j]$ is empty (it is possible only if $|Q_k|<|P_k|$)}}\;
                
            Update $\mathcal{Q}=\swap{\mathcal{Q}}{cl}{Q_k[j]}$\;
		    \label{step:swap2}
        } 
    } 
         
    } 
    \caption{
    \label{alg:simulate-bivalue}
    Transforming an MMS partition to an output of FFD}
\end{algorithm}

Let $\mathcal{P}$ be the output of FFD$(v,\tau)$. We prove that, by performing some swapping operations on $\mathcal{Q}$, we can make $\mathcal{Q}$ equal to $\mathcal{P}$. The swapping process is described as Algorithm \ref{alg:simulate-bivalue}. Below, we provide a detailed explanation of the process.

First note that, in the FFD outcome, the bundles are naturally ordered by descending order of the number of large chores, so that the first bundles contain $\lceil \tau / l \rceil$ large chores each, then there is possibly a single bundle that contains fewer large chores, and the remaining bundles contain only small chores.
In Line \ref{step:order}, we order the given MMS partition in the same way, by descending order of the large chores, and then by descending number of the small chores.

In the following lines, we process the bundles in $\mathcal{Q}$ and $\mathcal{P}$ in this order. For every $k$ in $1,\ldots, n$, we modify $\mathcal{Q}$ by swapping some chores, so that $Q_k$ becomes equal to $P_k$. Therefore, after the main loop ends, bundles $Q_1,\ldots,Q_n$ are equal to bundles $P_1,\ldots,P_n$, which means that in $\mathcal{P}$, all items are packed into $n$ bins.

We maintain the following invariants at the start of each iteration $k$:
\begin{itemize}
\item 
All bundles in $\mathcal{Q}$ processed in previous iterations are equal to their counterparts in $\mathcal{P}$, that is $Q_i=P_i$ for all $i<k$.
This holds vacuously in iteration $k=1$, and holds in the following iterations due to the swapping process.
\item 
All bundles in $\mathcal{Q}$ not processed yet have a cost of at most $\tau$. 
Moreover, each of these bundles satisfies one of the following conditions:
(a) A cost of at most $\vmms$, or ---
(b) No large chores at all, or ---
(c) One large chore and some number $b_i$ of small chores, such that $b_i+2$ small chores can fit into a single FFD bin.

Note that, in iteration $k=1$, invariant (a) holds, as the costs of all bundles in $\mathcal{Q}$ are at most $\vmms$ by definition of MMS.

Our main challenge will be to prove that the invariant is maintained in the following iterations.
\end{itemize}

At iteration $k$, transforming $Q_k$ to $P_k$ is done in 2 phases.

The first phase is the block starting at Line \ref{step:inner-loop-1}. It handles the quantity of chores in $Q_k$. If $Q_k$ contains more chores than $P_k$, then we reduce the number of chores in $Q_k$ by swapping two small chores from $Q_k$ with one large chore from a later bundle in $\mathcal{Q}$. To prove that this swap is possible, we prove below that $Q_k$ indeed contains two small chores, and that there is a large chore in a later bundle.

\begin{lemmarep}
\label{lem:twosmallchores}
In Algorithm~\ref{alg:simulate-bivalue} iteration $k$,
if $|Q_k| > |P_k|$, then 

(a) $P_k$ contains at least one large chore more than $Q_k$;

(b) $Q_k$ contains at least two small chores more than $P_k$;

(c) Some bundle $Q_z$ with $z>k$ contains a large chore (hence, $k<n$).

(d) $Q_k$ contains at least one large chore.
\end{lemmarep}
\begin{proof}
(a) 
By the loop invariant, $P_i = Q_i$ for all $i<k$, and $v(Q_k)\leq \tau$. Therefore, by FFD properties, $P_k \lex{\geq} Q_k$. As $Q_k$ contains more chores than $P_k$, $P_k$ must contain more large chores than $Q_k$.

(b)
As $P_k$ contains more large chores and $Q_k$ contains more chores overall, $Q_k$ must contain at least two small chores more than $P_k$.

(c) 
As $P_k$ contains some large chore that is not in $Q_k$, that large chore must be allocated elsewhere in $\mathcal{Q}$.
As $P_i = Q_i$ for all $i<k$, this large chore must be allocated in some $Q_z$ with $z>k$.

(d)
As the bundles in $\mathcal{Q}$ have been ordered by the number of large chores, $Q_k$ must contain at least as many large chores as $Q_z$.
\end{proof}

Note that, during the swap in Line \ref{step:swap1}, the cost of $Q_z$ might increase. We will prove later that, at the end of the iteration, it remains at most $\tau$, and also satisfies Invariant 3. Before that, we note two things:
\begin{itemize}
\item The algorithm only swaps a large chore from the \emph{last} bundle $Q_z$ containing a large chore.
Therefore, the order of bundles in $Q_k,\ldots,Q_n$ is maintained: we can still assume that the following bundles are ordered by descending order of the number of large chores.
\item 
The cost of the bundle $Q_z$ might grow above $\vmms$, but the costs of other bundles remain at most $\vmms$.
By Lemma \ref{lem:twosmallchores}(c), the last bundle $Q_z$ will never satisfy the condition $|Q_z|>|P_z|$, and will never arrive at Line \ref{step:swap1}.
Therefore, all bundles $Q_k$  that arrive at Line \ref{step:swap1} satisfy $v(Q_k)\leq \vmms$.
\end{itemize}

In most cases, $|Q_k|=|P_k|$  after a single swap, but in two special cases, another swap is needed. 
To analyze these cases, we divide to cases based on the numbers of large and small chores in $Q_k$ and $P_k$.
Due to Lemma \ref{lem:tau=H+s}, the total number of chores in $Q_k$ is at most $6$. 
Due to Lemma \ref{lem:twosmallchores}, 
if Line \ref{step:swap1} runs, 
then there are only $35$ cases for the numbers of chores; they are all listed in Figure \ref{fig:PkQk}.  
columns A--D.
For example, the first case (in row 3)
is the case in which $Q_k$ contains one large chore and two small chores, which is the minimum number required by Lemma \ref{lem:twosmallchores}.
In this case, by Lemma \ref{lem:twosmallchores}, $P_k$ must contain two large chores and no small chores.

\begin{figure*}
\begin{center}
\makebox[0pt]{
\includegraphics[width=0.9\textwidth]{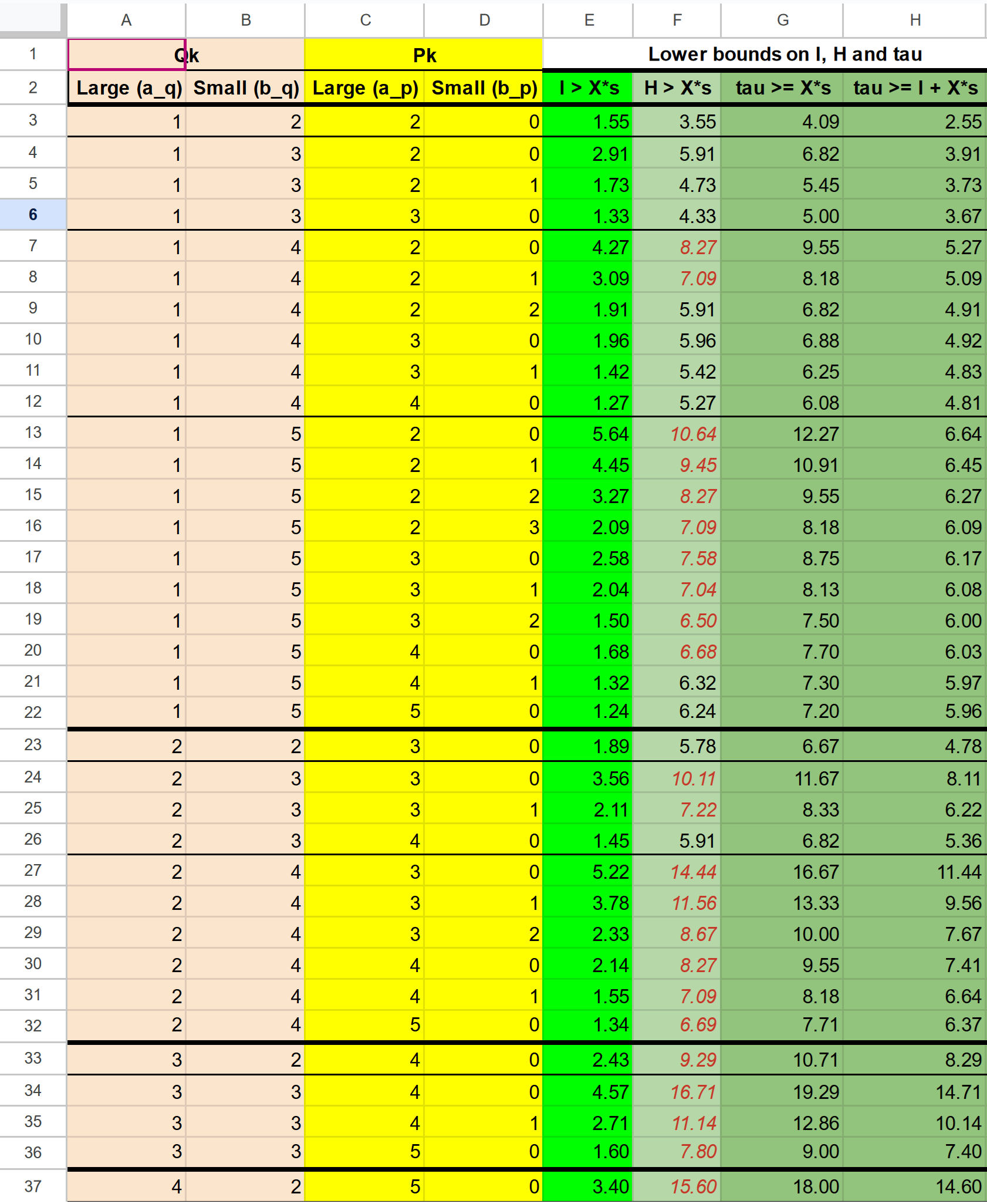}
}
\end{center}
\caption{
\label{fig:PkQk}
Columns A--D describe all possible cases for numbers of large and small chores in $P_k$ and $Q_k$, that satisfy Lemma \ref{lem:twosmallchores}.
\\
Column E describes a lower bound on $l/s$ derived from Lemma \ref{lem:PkQk}.
\\
Column F describes a lower bound on $\vmms/s$ derived from substituting column E in $\vmms \geq a_q l + b_q s$.
\\
Columns G and H 
describe lower bounds on $\tau$ derived from substituting column E in $\tau \geq (15/13)\cdot(a_q l + b_q s)$.
}
\end{figure*}

Given the numbers of large and small chores in $P_k$ and $Q_k$, we can derive a lower bound on the ratio $l/s$; that is, we derive an inequality of the kind $l > X\cdot s$, for some positive number $X$. 
The bound is proved below, and shown in Figure \ref{fig:PkQk} column E.

\begin{lemmarep}
\label{lem:PkQk}
Suppose $\tau = \frac{15}{13}\vmms$.
For any $a_q < a_p$ and $b_q > b_p$,
if $Q_k$ contains $a_q$ large and $b_q$ small chores,
and $P_k$ contains $a_p$ large and $b_p$ small chores, then
\begin{align*}
 &&
    l > (15 b_q - 13 b_p - 13)s / (13 a_p - 15 a_q).
\end{align*}
\end{lemmarep}
\begin{proof}
The condition on $Q_k$ implies
\begin{align*}
    a_q l + b_q s \leq \vmms = (13/15)\tau.
\end{align*}
On the other hand, $Q_k$ contains small chores that are unallocated in $P_k$, so by FFD properties, $P_k$ has no room for an additional small chore. Hence:
\begin{align*}
    a_p l + (b_p+1) s > \tau.
\end{align*}

Combining these inequalities gives:
\begin{align*}
    a_p l + (b_p+1) s &> \tau \geq (15/13)(a_q l + b_q s)
    \\
    (13 a_p - 15 a_q)l  &> (15 b_q - 13 b_p - 13)s
    \\
    l &> (15 b_q - 13 b_p - 13)s / (13 a_p - 15 a_q).
\end{align*}
This completes the proof.
\end{proof}

By substituting the lower bound on $l$ in the inequality 
$\vmms \geq a_q l + b_q s$, we can get a lower bound on the threshold $\vmms$ in terms $s$. That is, 
we can get an inequality of the form  $\vmms > X\cdot s$, for some positive number $X$. These bounds are shown in Figure \ref{fig:PkQk} column F. 
If the lower bound is at least $6.5s$, then by Corollary \ref{cor:6.5s} FFD packs all items into $n$ bins; therefore, we can ignore all the rows in which column F is at least $6.5$ (the relevant values in column F are italicized).

Similarly, by substituting the lower bound on $l$ in the inequality $\tau \geq (15/13)\cdot (a_q l + b_q s)$, we can get a lower bound on the threshold $\tau$ in terms of $l$ and $s$. In particular, we can get lower bounds of the forms: $\tau > X\cdot s$ or $\tau > l + X\cdot s$, for some positive numbers $X$. These upper bounds are shown in Figure \ref{fig:PkQk}, columns G and H.

Now we can see that, in almost all cases in which $\vmms < 6.5s$, the difference $|Q_k|-|P_k|=1$. There are  only two special cases, and they are handled by the inner condition starting at Line \ref{step:inner-condition-2}:
\begin{itemize}
\item $Q_k$ initially contained 1 large and 3 small chores, and $P_k$ contained 2 large and no small chores (row 4).
After the swap in Line \ref{step:swap1}, $Q_k$ contains 2 large and 1 small chore; we move the small chore to $Q_z$ and get $Q_k=P_k$.
\item 
$Q_k$ initially contained 1 large and 4 small chores, and $P_k$ contained 3 large and no small chores (row 10).
After the swap in Line \ref{step:swap1}, $Q_k$ contains 2 large and 2 small chores.
As it still contains fewer large chores than $P_k$, there must be some later bundle $Q_{z'}$ for $z'>k$ that contains a large chore; we do another swap of two small chores from $Q_k$ with one large chore from $Q_{z'}$, and get $Q_k=P_k$.
\end{itemize}

If one of the two special cases happens, $P_k = Q_k$ holds, and we can move to the next iteration.
Otherwise, $Q_k$ received at most one new large chore, and $|Q_k|\leq |P_k|$;
therefore,  $P_k \lex{\geq} Q_k$ holds.

Now we start the second phase of swaps, which is the loop starting at Line \ref{step:inner-loop-2}.
For convenience, we  add to $Q_k$ some dummy chores with cost $0$, so that $|Q_k| = |P_k|$. 
We process the chores in $Q_k$ and $P_k$ in order, from large to small. 
As  $P_k \lex{\geq} Q_k$, we always have $v(P_k[j])\geq v(Q_k[j])$.
If $v(P_k[j]) = v(Q_k[j])$ then there is nothing to do.
Otherwise, $v(P_k[j]) > v(Q_k[j])$,
which means that either (a) $P_k[j]$ is large and $Q_k[j]$ is small or dummy, or (b) $P_k[j]$ is small and $Q_k[j]$ is dummy.
In both cases, the chore $P_k[j]$ is missing from $Q_k[j]$. 
As for $i<k$ the bundles $Q_i$ and $P_i$ are identical, the missing chores must be found in some later bundle, say $Q_z$. 
We swap the smaller chore $Q_k[j]$ with the larger chore from $Q_z$.
The cost of $Q_z$ decreases, so the invariant is not affected.

Once the second loop ends $Q_k = P_k$,
so invariant 1 holds at the start of the next iteration. 

It remains to prove that invariants 2 and 3  hold too. For this we have to prove that, for every bundle $Q_z$ involved in a swap, its cost after the end of the iteration is at most $\tau$, and in addition, it satisfies one of the conditions (a) (b) or (c) in invariant 3.

Now we are ready to prove the loop invariants 2 and 3. We assume that the invariants hold at the iteration start, and prove that they hold at the iteration end.

\begin{lemmarep}    
If invariants 2 and 3 hold at the start of iteration $k$, then they also hold at the end of iteration $k$, so that the following holds for all $i>k$:
\begin{itemize}
\item $v(Q_i)\leq \tau$. 
\item Either (a) $v(Q_i)\leq \vmms$, or (b) $Q_i$ contains no large chores, or (c) $Q_i$ contains one large chore and some $b_i$ small chores and $(b_i+2)s\leq \tau$.
\end{itemize}
\end{lemmarep}  
\begin{proof}
If the conditions of Lemma \ref{lem:twosmallchores} are not satisfied in iteration $k$, then Line \ref{step:swap1} does not run, and thus the cost of $Q_i$ remains the same or decreases. Also, the number of large chores in $Q_i$ remains the same or decreases. Therefore the invariants still hold.

Otherwise, the numbers of chores in $Q_k$ and $P_k$ must be one of the cases in Figure \ref{fig:PkQk}, that is, one of the rows 3--37 where the value in column F is less than $6.5$. 
We have to prove that, in all these cases, the modified bundle $Q_z$ still satisfies the invariants. We consider three cases based on the contents of $Q_k$.

\paragraph{Case 1.}
$Q_k$ contains exactly one large chore, and some $b_q$ small chores (rows 3--22 in the figure).
By the ordering on $\mathcal{Q}$, $Q_z$ contains exactly one large chore, and at most $b_q$ small chores.

After the swap in Line \ref{step:swap1}, $Q_z$ contains no large chores and at most $b_q+2$ small chores, so $v(Q_z)\leq (b_q+2)s$.
By looking at Column G it can be verified that, in all rows 3--22, $\tau \geq (b_q+2)\cdot s$. Therefore, $v(Q_z)\leq \tau$ in all these cases, and $Q_z$ satisfies invariant (b).

We now handle the two special cases, where another swap happens.
\begin{itemize}
\item In the special case of row 4 (where $Q_k$ initially contains 1 large and 3 small chores), $\tau > 6s$. Therefore, $v(Q_z) < \tau$ even after one more small chore is moved.
\item In the special case of row 10 (where $Q_k$ initially contains 1 large and 4 small chores), 
$Q_z$ after the first swap contains no large chores, so the second swap will occur with another bundle $Q_{z'}$.
Similarly to $Q_z$, this bundle too contains 1 large and at most 4 small chores, so after the swap it contains at most 6 small chores, so $v(Q_{z'}) \leq 6s < \tau$.
\end{itemize}

\paragraph{Case 2.}
$Q_k$ contains 2 large and 2 small chores,
and $P_k$ contains 3 large chores (row 23 in the figure). 
In this case $\tau > 6.67s$ (see column G).
We now split to sub-cases based on the number of large chores in $Q_z$.

\paragraph{Subcase 2.1(a).}
$Q_z$ contains 1 large chore and satisfies invariant (a), that is, $v(Q_z)\leq \vmms$. 
This means that $Q_z$ contains some $b_z$ small chores such that $l  + b_z s \leq \vmms$.
After the swap, $Q_z$ contains no large chores and $(b_z+2)$ small chores.
If $b_z\leq 4$, then $v(Q_z)\leq 6s < \tau$,
so $Q_z$ satisfies invariant (b).
Otherwise, $l+5s\leq \vmms$.
As $P_k$ contains 3 large chores, 
the situation is as in row 17, where $7.58s <\vmms$, so the case is handled by Corollary \ref{cor:6.5s}.

\paragraph{Subcase 2.1(c).}
$Q_z$ contains 1 large chore and satisfies invariant (c), that is, it contains some $b_z$ small chores such that $(b_z+2)s\leq \tau$. 
After the swap, $Q_z$ contains no large chores and $(b_z+2)$ small chores, 
so $v(Q_z)\leq \tau$, and $Q_z$ satisfies invariant (b).

\paragraph{Subcase 2.2.}
$Q_z$ contains 2 large chores. By the ordering on $\mathcal{Q}$, it contains at most 2 small chores. 
After the swap in Line \ref{step:swap1}, $Q_z$ contains $1$ large chore and at most $4$ small chores.
In column H it can be seen that $\tau \geq l + 4\cdot s$, so  $v(Q_z)\leq \tau$.
So  $v(Q_z)\leq 6s < \tau$, and $Q_z$ satisfies invariant (b).


\paragraph{Case 3.}
$Q_k$ contains 2 large and 3 small chores,
and $P_k$ contains 4 large chores (row 26 in the figure). 
In this case $\tau > 6.82s$ (see column G).

Again we split to sub-cases based on the number of large chores in $Q_z$.

\paragraph{Subcase 3.1(a).}
$Q_z$ contains 1 large chore and satisfies invariant (a), that is, $v(Q_z)\leq \vmms$. 
This means that $Q_z$ contains some $b_z$ small chores such that $l  + b_z s \leq \vmms$.
After the swap, $Q_z$ contains no large chores and $(b_z+2)$ small chores.
If $b_z\leq 4$, then $v(Q_z)\leq 6s < \tau$,
so $Q_z$ satisfies invariant (b).
Otherwise, $l+5s\leq \vmms$.
As $P_k$ contains 4 large chores,
the situation is as in row 20, where $6.68s <\vmms$, so the case is handled by Corollary \ref{cor:6.5s}.

\paragraph{Subcase 3.1(c).}
This case is handled exactly like Subcase 2.1(c) above.

\paragraph{Subcase 3.2.}
$Q_z$ contains 2 large chores. By the ordering on $\mathcal{Q}$, it contains at most 3 small chores. 
After the swap in Line \ref{step:swap1}, $Q_z$ contains $1$ large chore and at most $5$ small chores.
At the same time, $Q_k$ contains 3 large chores and one 1 small chores.
As a result, $|Q_k|=|P_k|$ but $Q_k\neq P_k$, so the algorithm moves to the second loop. 
In Line \ref{step:swap2}, 
the algorithm swaps one small chore from $Q_k$ with one large chore from $Q_z$. 
At the end of the iteration, $Q_z$ contains no large chores and at most $6$ small chores.
So  $v(Q_z)\leq 6s < \tau$, and $Q_z$ satisfies invariant (b).
\end{proof}

We have proved that the invariants hold throughout Algorithm~\ref{alg:simulate-bivalue}. Therefore, they also hold when the algorithm ends. Therefore, when the algorithm ends, $Q_k=P_k$ for all $k\in\{1,\ldots,n\}$. This shows that FFD has packed all chores in $\mathcal{Q}$ into $n$ bundles. Combined with Example \ref{exm:15/13}, we get the tight bound declared in Theorem \ref{thm:bivalued}.

\section{Conclusion}
\citet{HuangL21} first explored the connection between the bin packing problem and the fair allocation of chores, an approach that was further utilized in~\citep{huang2023reduction}. Building on this foundation, our paper delves deeper into this connection. We demonstrated that the HFFD algorithm could improve the state-of-the-art of maximin share (MMS) in factored instances, personalized bivalued instances, and 1-out-of-$d$ MMS allocations. The central question we explore is: Can we discover additional connections between bin packing and the fair allocation of chores?

An intriguing open question in this area is whether the HFFD (and FFD) algorithm is monotonic in a certain range. In this paper, we demonstrate the monotonicity of the algorithm for two special classes of chores costs. This question is closely tied to the possibility of designing an efficient algorithm for MMS allocations that achieves the exact approximation ratio of the bin packing problem. The monotonicity in the general case remains unclear. Another open problem is whether MMS allocations always exist for personalized bivalued instances. To the best of our knowledge, there are no known counterexamples demonstrating the non-existence of MMS for these instances. 

\section*{Acknowledgments}
Jugal Garg was supported by NSF Grants CCF-1942321 and CCF-2334461. Xin Huang was supported by JST ERATO Grant Number JPMJER2301, Japan. Erel Segal-Halevi was funded by Israel Science Foundation grant no. 712/20.

\bibliographystyle{plainnat}
\bibliography{aaai25}

\appendix
\end{document}